\documentclass[submission,copyright,creativecommons]{eptcs}
\providecommand{\event}{LFMTP 2024} 

\usepackage{iftex}

\ifpdf
  \usepackage{underscore}         
  \usepackage[T1]{fontenc}        
\else
  \usepackage{breakurl}           
\fi

\title{Kuroda's Translation for the\\ $\lambda\Pi$-Calculus Modulo Theory and Dedukti}
\author{Thomas Traversié
\institute{Université Paris-Saclay, CentraleSupélec, MICS \\ Gif-sur-Yvette, France}
\institute{Université Paris-Saclay, Inria, CNRS, ENS-Paris-Saclay, LMF \\ Gif-sur-Yvette, France}
\email{thomas.traversie@centralesupelec.fr}}
\def\titlerunning{Kuroda's Translation for the $\lambda\Pi$-Calculus Modulo Theory and Dedukti}
\def\authorrunning{T. Traversié}

\usepackage{xspace}
\usepackage{amsmath}
\usepackage{amssymb}
\usepackage{amsthm}
\usepackage{mathpartir}
\usepackage{mathtools}
\usepackage[capitalise]{cleveref}
\usepackage{stmaryrd}
\usepackage{enumitem}

\newtheorem{definition}{Definition}
\newtheorem{lemma}{Lemma}
\newtheorem{proposition}{Proposition}
\newtheorem{theorem}{Theorem}
\newtheorem{example}{Example}

\def\imp{\mathbin{\Rightarrow}}
\def\conj{\mathbin{\wedge}}
\def\disj{\mathbin{\vee}}
\def\ex{{\exists}}
\def\fa{{\forall}}

\def\eqi{\Leftrightarrow}
\def\ra{\rightarrow}
\def\lra{\hookrightarrow}
\def\Type{\mbox{\tt TYPE}}
\def\Kind{\mbox{\tt KIND}}
\def\Set{{\it Set}}
\def\El{{\it El}}
\def\Prop{{\it Prop}}
\def\Prf{{\it Prf}}
\def\imp{\mathbin{\Rightarrow}}

\def\arr{\mathbin{\rightsquigarrow}}
\def\o{o}

\def\T{\mathcal{T}}
\def\R{\mathcal{R}}

\newcommand\ck{{\sc Construkti}\xspace}
\newcommand\dk{{\sc Dedukti}\xspace}

\newcommand{\lpc}{$\lambda \Pi$-calculus\xspace}
\newcommand{\lpcm}{$\lambda \Pi$-calculus modulo theory\xspace}

\usepackage{xcolor}
\usepackage{listings}
\definecolor{green}{RGB}{0,130,0}
\definecolor{lightgrey}{RGB}{240,240,240}
\lstdefinelanguage{Dedukti}
{
  inputencoding=utf8,
  extendedchars=true,
  numbers=none,
  numberstyle={},
  tabsize=2,
  basicstyle={\ttfamily\upshape\mdseries},
  backgroundcolor=\color{lightgrey},
  keywords={def,thm,injective},
  sensitive=true,
  keywordstyle=\color{blue},
  morecomment=[s]{(;}{;)},
  commentstyle={\itshape\color{red}},
  string=[b]{"},
  stringstyle=\color{orange},
  showstringspaces=false,
}
\begin{document}

\maketitle

\begin{abstract}
Kuroda's translation embeds classical first-order logic into intuitionistic logic, through the insertion of double negations. Recently, Brown and Rizkallah extended this translation to higher-order logic. In this paper, we adapt it for theories encoded in higher-order logic in the $\lambda\Pi$-calculus modulo theory, a logical framework that extends $\lambda$-calculus with dependent types and user-defined rewrite rules. We develop a tool that implements Kuroda's translation for proofs written in {\sc Dedukti}, a proof language based on the $\lambda\Pi$-calculus modulo theory.
\end{abstract}

\section{Introduction}

The \lpcm~\cite{lambdapi} is an extension of simply typed $\lambda$-calculus with dependent types and user-defined rewrite rules. It is a logical framework, meaning that one can express many theories in it---through the definitions of typed constants and rewrite rules. For instance, it is possible to encode Predicate Logic, Simple Type Theory and the Calculus of Constructions in the \lpcm~\cite{theoryU}. In particular, theories from other proof systems can be expressed inside this logical framework~\cite{thire}. The \lpcm has been implemented in the concrete language \dk~\cite{expressing,deduktiengine}. Besides automatic proof checking, \dk can be used as a common language to exchange proofs between different systems. However, if one wants to translate proofs from the \textit{classical} proof assistant \textsc{HOL Light} to the \textit{intuitionistic} proof assistant \textsc{Coq} \textit{via} \dk, one must transform classical proofs into intuitionistic proofs \textit{inside} \dk.

Classical logic corresponds to intuitionistic logic extended with the principle of excluded middle $A \disj \neg A$, or equivalently the double-negation elimination $\neg\neg A \imp A$. Classical logic can be embedded into intuitionistic logic, using double-negations translations. Glivenko~\cite{glivenko1928} proved that any propositional formula $A$ is provable in classical logic if and only if its double negation $\neg\neg A$ is provable in intuitionistic logic. Kolmogorov~\cite{kolmogorov1925}, Gödel~\cite{godel1933}, Gentzen~\cite{gentzen1936} and Kuroda~\cite{kuroda1951} developed double-negation translations $A \mapsto A^*$, which transforms any first-order formula $A$ such that:
\begin{enumerate}[label=(\roman*)]
\item if $A$ is provable in classical logic then its translation $A^*$ is provable in intuitionistic logic, \label{property_1}
\item $A$ and $A^*$ are classically equivalent. \label{property_2}
\end{enumerate}
More recently, Brown and Rizkallah~\cite{brown_rizkallah} showed that Kolmogorov's and Gödel-Gentzen's translations cannot be extended to higher-order logic. They proved that, in higher-order logic, Kuroda's translation satisfies Property~\ref{property_1}, but that it fails in the presence of functional extensionality. In fact~\cite{kuroda_hol}, Property~\ref{property_1} holds in the presence of functional extensionality under some specific condition, and Property~\ref{property_2} holds when assuming functional extensionality and propositional extensionality.

\paragraph{Contribution.} In this paper, we express Kuroda's translation for theories of the \lpcm that are encoded in higher-order logic. It is both an encoding---into a logical framework that features proofs as terms---and an extension---to a logical framework that features dependent types and user-defined rewrite rules---of Kuroda's translation. We implement such translation inside \ck, a tool that translates \dk files. \ck is tested on a benchmark of a hundred formal proofs. This tool and this benchmark are available at \url{https://github.com/Deducteam/Construkti}.

\paragraph{Outline of the paper.} In \Cref{sec_lpcm}, we present the \lpcm and we detail an encoding of higher-order logic in it. In \Cref{sec_kuroda}, we define Kuroda's translation for theories of \lpcm that are encoded in higher-order logic, and we prove the embedding of classical logic into intuitionistic logic. In \Cref{sec_implem}, we implement \ck and test it on \dk proofs.

\section{Higher-Order Logic in the \texorpdfstring{$\lambda\Pi$}{lambdaPi}-Calculus Modulo Theory}
\label{sec_lpcm}

In this section, we present the \lpcm, and we detail an encoding of higher-order logic in this logical framework. We characterize the theories considered in the rest of this paper---theories encoded in higher-order logic.

\subsection{The \texorpdfstring{$\lambda\Pi$}{lambdaPi}-Calculus Modulo Theory}

The Edinburgh Logical Framework~\cite{LF}, also called \lpc, is an extension of simply typed $\lambda$-calculus with dependent types. The \lpcm~\cite{lambdapi} corresponds to the Edinburgh Logical Framework extended with user-defined rewrite rules~\cite{rewriteSystem}. Its syntax is given by:
\begin{align*}
&\text{\textit{Sorts}} &s &\Coloneqq \Type ~|~ \Kind \\
&\text{\textit{Terms}} &t,u, A, B &\Coloneqq c ~|~ x ~|~ s ~|~ \Pi x : A. ~B ~|~ \lambda x : A. ~t ~|~ t ~u \\
&\text{\textit{Contexts}} &\Gamma &\Coloneqq \langle \rangle ~|~ \Gamma, x : A \\
&\text{\textit{Signatures}} &\Sigma &\Coloneqq \langle \rangle ~|~ \Sigma, c : A \\
&\text{\textit{Rewrite systems}} &\R &\Coloneqq \langle \rangle ~|~ \R, \ell \lra r
\end{align*}
where $c$ is a constant and $x$ is a variable (ranging over disjoint sets). $\Type$ and $\Kind$ are two sorts: terms of type $\Type$ are called types, and terms of type $\Kind$ are called kinds. $\Pi x : A. ~B$ is a dependent product (simply written $A \ra B$ if $x$ does not occur in $B$), $\lambda x : A. ~t$ is an abstraction, and $t ~u$ is an application. Contexts, signatures and rewrite systems are finite sequences, and are written $\langle \rangle$ when empty. Signatures $\Sigma$ are composed of typed constants $c : A$, where $A$ is a closed term (that is a term with no free variables). Rewrite systems $\R$ are composed of rewrite rules $\ell \lra r$, where the head symbol of $\ell$ is a constant. The \lpcm is a logical framework, in which $\Sigma$ and $\R$ are fixed by the users depending on the logic they are working in. The relation $\lra_{\beta\R}$ is generated by $\beta$-reduction and by the rewrite rules of $\R$. The conversion $\equiv_{\beta\R}$ is the reflexive, symmetric, and transitive closure of $\lra_{\beta\R}$.

The typing rules for the \lpcm are given in \Cref{typrules_lpcm}. We write $\vdash \Gamma$ when the context $\Gamma$ is well formed, and $\Gamma \vdash t : A$ when the term $t$ is of type $A$ in the context $\Gamma$. For convenience, $\langle \rangle \vdash t : A$ is simply written $\vdash t : A$. The standard weakening rule is admissible.

\begin{figure}
\begin{mathpar}
\inferrule*[right={[Empty]}]{ }{\vdash \langle \rangle}

\inferrule*[right={[Decl] $x \notin \Gamma$}]{\vdash \Gamma \\ \Gamma \vdash A : s}{\vdash \Gamma, x : A}

\inferrule*[right={[Sort]}]{\vdash \Gamma}{\Gamma \vdash \Type : \Kind}

\inferrule*[right={[Const] $c : A \in \Sigma$}]{\vdash \Gamma \\ \vdash A : s}{\Gamma \vdash c : A}

\inferrule*[right={[Var] $x : A \in \Gamma$}]{\vdash \Gamma}{\Gamma \vdash x : A}

\inferrule*[right={[Prod]}]{\Gamma \vdash A : \Type \\ \Gamma, x : A \vdash B : s}{\Gamma \vdash \Pi x : A. ~B : s}

\inferrule*[right={[Abs]}]{\Gamma \vdash A : \Type \\ \Gamma, x : A \vdash B : s \\ \Gamma, x : A \vdash t : B}{\Gamma \vdash \lambda x : A. ~t : \Pi x : A. ~B}

\inferrule*[right={[App]}]{\Gamma \vdash t : \Pi x : A. ~B \\ \Gamma \vdash u : A}{\Gamma \vdash t\ u : B[x \leftarrow u]}

\inferrule*[right={[Conv] $A \equiv_{\beta\R} B$}]{\Gamma \vdash t : A \\ \Gamma \vdash B : s}{\Gamma \vdash t : B}
\end{mathpar}
\caption{Typing rules of the \lpcm.}
\label{typrules_lpcm}
\end{figure}

We write $\Lambda(\Sigma)$ for the set of terms whose constants belong to $\Sigma$. We say that ($\Sigma, \R)$ is a theory when: $(i)$ for each rule $\ell \lra r \in \R$, both $\ell$ and $r$ belongs to $\Lambda(\Sigma)$, $(ii)$ $\lra_{\beta\R}$ is confluent on $\Lambda(\Sigma)$, and $(iii)$ each rule $\ell \lra r \in \R$ preserves types (for all context $\Gamma$, substitution $\theta$, and term $A \in \Lambda(\Sigma)$, if $\Gamma \vdash \ell\theta : A$ then $\Gamma \vdash r\theta : A$).

In the \lpcm, if $\Gamma \vdash t : A$ then $\Gamma$ is well-formed and $A$ is well-typed. To prove this, we use the two following properties.

\begin{lemma}
If $\Gamma \vdash t : A$, then either $A = \Kind$ or $\Gamma \vdash A : s$ for $s = \Type$ or $s = \Kind$. If $\Gamma \vdash \Pi x : A. ~B : s$, then $\Gamma \vdash A : \Type$.
\end{lemma}

\subsection{An Encoding of Higher-Order Logic}

It is possible to express higher-order logic in the \lpcm~\cite{theoryU}. For this, we have to introduce the notions of proposition and proof. We declare the constant $\Set$, which represents the universe of sorts, along with the injection $\El$ that maps sorts to the type of its elements. The constant $\Prop$ defines the universe of propositions, and the injection $\Prf$ maps propositions into the type of its proofs. In this encoding, we say that $P$ of type $\Prop$ is a proposition, that $\Prf ~P$ is a formula and that a term of type $\Prf ~P$ is a proof of $P$.
\begin{flalign*}
&\Set : \Type & &\El : \Set \ra \Type & &\arr : \Set \ra \Set \ra \Set & &\o : \Set &\\
&\Prop : \Type & &\Prf : \Prop \ra \Type & &\El ~(x \arr y) \lra \El ~x \ra \El ~y & &\El ~\o \lra \Prop &
\end{flalign*}
The arrow $\arr$ (written infix) is used to represent function types between terms of type $\Set$. Propositions are considered as objects, using the sort $\o$ and the rewrite rule $\El ~\o \lra \Prop$.

Now that we have introduced the notions of proposition and proof, we can define the logical connectives and quantifiers of predicate logic.
\begin{flalign*}
&\imp : \Prop \ra \Prop \ra \Prop & &\top : \Prop & &\fa : \Pi x : \Set. ~(\El ~x \ra \Prop) \ra \Prop &\\
&\conj : \Prop \ra \Prop \ra \Prop  & &\bot : \Prop & &\ex : \Pi x : \Set. ~(\El ~x \ra \Prop) \ra \Prop &\\
&\disj : \Prop \ra \Prop \ra \Prop & &\neg : \Prop \ra \Prop & &\eqi ~: \Prop \ra \Prop \ra \Prop &
\end{flalign*}
Remark that $\fa$ and $\ex$ are polymorphic quantifiers that can be applied to the sort of proposition $o$. Hence the higher-order feature directly derives from the rewrite rule $\El ~\o \lra \Prop$.

In natural deduction, each connective and quantifier comes with an introduction and an elimination inference rule. The encoding of the notions of proposition and proof is well-suited for representing inference rules: logical consequences are represented by arrow types, and parameters are represented by dependent types. For instance, the inference rule for the elimination of disjunction
\begin{mathpar}
\inferrule*{\Gamma \vdash P \disj Q \\ \Gamma, P \vdash R \\ \Gamma, Q \vdash R}{\Gamma \vdash R}
\end{mathpar}
is simply expressed by the constant $\mathsf{or_e}$ of type 
$$\Pi p,q : \Prop. ~\Prf ~(p \disj q) \ra \Pi r : \Prop. ~(\Prf ~p \ra \Prf ~r) \ra (\Prf ~q \ra \Prf ~r) \ra \Prf ~r$$
that can be used for any context $\Gamma$. The constants representing the natural deduction rules for the logical connectives are:
\begin{flalign*}
&\mathsf{imp_i} : \Pi p,q : \Prop. ~(\Prf ~p \ra \Prf ~q) \ra \Prf ~(p \imp q) &\\
&\mathsf{imp_e} : \Pi p,q : \Prop. ~\Prf ~(p \imp q) \ra \Prf ~p \ra \Prf ~q & \\
&\mathsf{and_i} : \Pi p : \Prop. ~\Prf ~p \ra \Pi q : \Prop. ~\Prf ~q \ra \Prf ~(p \conj q) & \\
&\mathsf{and_{e\ell}} : \Pi p,q : \Prop. ~\Prf ~(p \conj q) \ra \Prf ~p & \\
&\mathsf{and_{er}} : \Pi p,q : \Prop. ~\Prf ~(p \conj q) \ra \Prf ~q & \\
&\mathsf{or_{i\ell}} : \Pi p : \Prop. ~\Prf ~p \ra \Pi q : \Prop. ~\Prf ~(p \disj q) & \\
&\mathsf{or_{ir}} : \Pi p,q : \Prop. ~\Prf ~q \ra \Prf ~(p \disj q) & \\
&\mathsf{or_e} : \Pi p,q : \Prop. ~\Prf ~(p \disj q) \ra \Pi r : \Prop. ~(\Prf ~p \ra \Prf ~r) \ra (\Prf ~q \ra \Prf ~r) \ra \Prf ~r & \\
&\mathsf{neg_i} : \Pi p : \Prop. ~(\Prf ~p \ra \Prf ~\bot) \ra \Prf ~(\neg p) & \\
&\mathsf{neg_e} : \Pi p : \Prop. ~\Prf ~(\neg p) \ra \Prf ~p \ra \Prf ~\bot & 
\end{flalign*}
For convenience, the semantic of the logical biconditional is encoded through the rewrite rule $p \eqi q \lra (p \imp q) \conj (q \imp p)$. The introduction of tautology and the elimination of contradiction are encoded by: 
\begin{flalign*}
&\mathsf{top_i} : \Prf ~\top & \\
&\mathsf{bot_e} : \Prf ~\bot \ra \Pi p : \Prop. ~\Prf ~p &
\end{flalign*}
The natural deduction rules for the quantifiers are represented by the following constants:
\begin{flalign*}
&\mathsf{all_i} : \Pi a : \Set. ~\Pi p : \El ~a \ra \Prop. ~(\Pi x : \El ~a. ~\Prf ~(p ~x)) \ra \Prf ~(\fa ~a ~p) & \\
&\mathsf{all_e} : \Pi a : \Set. ~\Pi p : \El ~a \ra \Prop. ~\Prf ~(\fa ~a ~p) \ra \Pi x : \El ~a. ~\Prf ~(p ~x) & \\
&\mathsf{ex_i} : \Pi a : \Set. ~\Pi p : \El ~a \ra \Prop. ~\Pi x : \El ~ a. ~\Prf ~(p ~x) \ra \Prf ~(\ex ~a ~p) & \\
&\mathsf{ex_e} : \Pi a : \Set. ~\Pi p : \El ~a \ra \Prop. ~\Prf ~(\ex ~a ~p) \ra \Pi r : \Prop. ~(\Pi x : \El ~a. ~\Prf ~(p ~x) \ra \Prf ~r) \ra \Prf ~r & 
\end{flalign*}
All those constants and rewrite rules define the encoding of intuitionistic higher-order logic in the \lpcm. We write $\Sigma_{HOL}^i$ for its constants and $\R_{HOL}$ for its rewrite rules. The principle of excluded middle is represented by:
\begin{flalign*}
&\mathsf{pem} : \Pi p : \Prop. ~\Prf ~(p \disj \neg p) & 
\end{flalign*}
Classical higher-order logic is encoded in the \lpcm by the constants $\Sigma_{HOL}^c$ (that is $\Sigma_{HOL}^i$ along with $\mathsf{pem}$) and by the rewrite rules $\R_{HOL}$.

Remark that we have decided to encode the natural deduction rules via \textit{typed constants}, while they are often expressed via \textit{rewrite rules} in the \lpcm~\cite{theoryU}. For instance, both the introduction and the elimination of implication can be derived from the rewrite rule $\Prf ~(p \imp q) \lra \Prf ~p \ra \Prf ~q$. So as to perform the translation from classical logic to intuitionistic logic, the natural deduction steps must be \textit{explicit} deduction steps, and cannot be \textit{implicit} computation steps. That is why we encode the natural deduction rules with a deep embedding---via typed constants---instead of a shallow embedding---via rewrite rules.

\subsection{Theories Encoded in Higher-Order Logic}

When working with the encoding of higher-order logic in the \lpcm, it is possible to mix sorts, propositions and proofs---which is not expected in higher-order logic. For example, propositions can be inserted in sorts when we have a term of type $\Prop \ra \Set$, and proofs can be inserted in propositions when we have a term of type $\Pi p : \Prop. ~\Prf ~p \ra \Prop$. To avoid such behavior, we introduce five grammars:
\begin{align*}
\kappa_1 &\Coloneqq \Set \mid \kappa_1 \ra \kappa_1 \\
\kappa_2 &\Coloneqq \Prop \mid \El ~a \mid \Pi x : \kappa_i. ~\kappa_2 \text{ with $i \in \{ 1, 2 \}$} \\
\kappa_3 &\Coloneqq \Prf ~p \mid \kappa_3 \ra \kappa_3 \mid \Pi x : \kappa_i. ~\kappa_3 \text{ with $i \in \{ 1, 2 \}$} \\
\kappa_4 &\Coloneqq \Type \mid \Pi x : \kappa_i. ~\kappa_4 \text{ with $i \in \{ 1, 2 \}$} \\
\kappa_5 &\Coloneqq \Kind
\end{align*}
The grammar $\kappa_3$ generates formulas and inference rules. The grammar $\kappa_4$ generates a subclass of kinds, and $\kappa_5$ only generates $\Kind$. We characterize the judgments of the \lpcm to ensure that types and kinds are generated by one of those grammars.

\begin{definition}[$\kappa$-property]
The judgment $\Gamma \vdash t : A$ satisfies the $\kappa$-property when $A \in \kappa_i$ for some $i \in \llbracket 1, 5 \rrbracket$. The judgment $\vdash \Gamma$ satisfies the $\kappa$-property when for each $(x : A) \in \Gamma$ we have $A \in \kappa_i$ for some $i \in \llbracket 1, 5 \rrbracket$. A derivation satisfies the $\kappa$-property when each of its judgments satisfies the $\kappa$-property.
\end{definition}
Theories encoded in higher-order logic are theories that feature the base higher-order encoding and in which the user-defined constants satisfy the $\kappa$-property.

\begin{definition}[Theory encoded in higher-order logic]
Let $\T = (\Sigma, \R)$ be a theory in the \lpcm. $\T$ is encoded in higher-order logic when:
\begin{enumerate}
    \item $\Sigma = \Sigma_{HOL}^k \cup \Sigma_{\T}$ with $k \in \{ i, c \}$ and $\Sigma_{HOL} \cap \Sigma_{\T} = \emptyset$,
    \item $\R = \R_{HOL} \cup \R_{\T}$ with $\R_{HOL} \cap \R_{\T} = \emptyset$,
    \item for every $c : A \in \Sigma_{\T}$, the judgment $\vdash c : A$ satisfies the $\kappa$-property,
    \item for every $\ell \lra r \in \R_{\T}$, $\ell$ is neither $\Prf$ nor $\fa$.
\end{enumerate}
\end{definition}
The fourth condition will ensure that the translation of a rewrite rule is a well-defined rewrite rule. Theories encoded in higher-order logic extend higher-order logic with user-defined rewrite rules and inference rules. The introduction of rewrite rules is part and parcel of deduction modulo theory~\cite{deduction_modulo}, while the introduction of inference rules has been developed in superdeduction modulo theory~\cite{superdeduction,superdeduction_phd}.

When considering a theory encoded in higher-order logic, all the user-defined constants satisfy the $\kappa$-property. In that respect, the only way to mix sorts, propositions and proofs is through $\lambda$-abstractions. For instance, $(\lambda P : \Prop. ~\o)$ is a term taking as input a proposition and returning a sort. The type $\El ~((\lambda P : \Prop. ~\o) ~\bot)$ mixes propositions and sorts, but it is $\beta$-convertible to $\El ~\o$, in which no proposition occurs. Using this principle, we can transform every derivation of a theory encoded in higher-order logic into a derivation that satisfies the $\kappa$-property, by applying $\beta$-reduction on fragments of the derivation. When a derivation satisfies the $\kappa$-property, the rewrite rules $\ell \lra r$ with $\ell$ and $r$ of type $A \in \kappa_3$ cannot be used. In the rest of this paper and without loss of generality, we only consider derivations that satisfy the $\kappa$-property and rewrite rules $\ell \lra r$ with $\ell$ and $r$ of type $A \in \kappa_i$ for $i \neq 3$.

\begin{example}[Equational theory]
\label{ex_equality}
Consider the theory $\T = (\Sigma_{HOL} \cup \Sigma_{eq}, \R_{HOL} \cup \R_{eq})$, with a polymorphic equality symbol $= ~: \Pi a : \Set. ~\El ~a \ra \El ~a \ra \Prop$, and a rewrite rule for the Leibniz principle $\Prf ~(= a ~x ~y) \lra \Pi P : \El ~a \ra \Prop. ~\Prf ~(P ~x) \ra \Prf ~(P ~y)$. This theory is encoded in higher-order logic. We can prove that the equality is reflexive, symmetric and transitive. For instance, the proof of reflexivity is given by $\lambda a : \Set. ~\mathsf{all_i} ~a ~(\lambda x : \El ~a. = a ~x ~x) ~(\lambda x : \El ~a. ~\lambda P : \El ~a \ra \Prop. ~\lambda P_x : \Prf ~(P ~x). ~P_x)$ which is of type $\Pi a : \Set. ~\Prf ~(\fa ~a ~(\lambda x : \El ~a. = a ~x ~x))$. 
\end{example}

\section{Kuroda's Translation in the \texorpdfstring{$\lambda\Pi$}{lambdaPi}-Calculus Modulo Theory}
\label{sec_kuroda}

In this section, we adapt Kuroda's double-negation translation to the \lpcm, when working in theories encoded in higher-order logic. Kuroda's translation~\cite{kuroda1951} inserts a double negation in front of formulas and one after every universal quantifier. More formally, we have $A^{Ku} \coloneqq \neg\neg A_{Ku}$ where $A_{Ku}$ is defined by induction:
\[
  \begin{array}{llll}
        (A \imp B)_{Ku} \coloneqq A_{Ku} \imp B_{Ku} &(\neg A)_{Ku} \coloneqq \neg A_{Ku} &P_{Ku} \coloneqq P \text{ if $P$ atomic} \\
        (A \conj B)_{Ku} \coloneqq A_{Ku} \conj B_{Ku} &\top_{Ku} \coloneqq \top &(\fa x ~A)_{Ku} \coloneqq \fa x ~\neg\neg A_{Ku} \\
        (A \disj B)_{Ku} \coloneqq A_{Ku} \disj B_{Ku} &\bot_{Ku} \coloneqq \bot &(\ex x ~A)_{Ku} \coloneqq \ex x ~A_{Ku} \\
  \end{array}
\]
This translation embeds classical logic into intuitionistic logic, as for any first-order formula $A$ we have $\Gamma \vdash A$ in classical logic if and only if $\Gamma^{Ku} \vdash A^{Ku}$ in intuitionistic logic.

\subsection{Translation of Terms and Theories}

When working inside a theory encoded in higher-order logic in the \lpcm, every formula has head symbol $\Prf$. Inserting a double negation in front of every formula is therefore equivalent to inserting it after every $\Prf$ symbol. In that respect, we define a single translation $t \mapsto t^{Ku}$ by induction on the terms of the \lpcm. The translation of $\Prf$ is $\lambda p. ~\Prf ~(\neg\neg p)$, and the translation of the universal quantifier $\fa$ is $\lambda a. ~\lambda p. ~\fa ~a ~(\lambda z. ~\neg\neg (p ~z))$. The translation of $\lambda$-abstraction $\lambda x : A. ~t$ is naturally given by $\lambda x : A^{Ku}. ~t^{Ku}$, the one of dependent type $\Pi x : A. ~B$ is given by $\Pi x : A^{Ku}. ~B^{Ku}$ and the one of application $t ~u$ is defined by $t^{Ku} ~u^{Ku}$. 

As we are in the \lpcm with the \textit{proofs-as-terms} paradigm, we have to translate proofs as well. Kuroda's translation relies on the fact that the translation of each natural deduction rule is admissible in intuitionistic logic. For instance, the introduction of implication allows to derive $\Gamma \vdash P \imp Q$ from $\Gamma, P \vdash Q$. In intuitionistic logic, $\Gamma^{Ku} \vdash (P \imp Q)^{Ku}$ is derivable from $\Gamma^{Ku}, P^{Ku} \vdash Q^{Ku}$. In the \lpcm, the \textit{constant} $\mathsf{imp_i}$ is of type $\Pi p,q : \Prop. ~(\Prf ~p \ra \Prf ~q) \ra \Prf ~(p \imp q)$, and we can build a \textit{term} $\mathsf{imp_i}^i$ of type $\Pi p,q : \Prop. ~(\Prf ~\neg\neg p \ra \Prf ~\neg\neg q) \ra \Prf ~\neg\neg (p \imp q)$, that only depends on the constants representing \textit{intuitionistic} natural deduction rules. Each constant $c$ of type $A$ representing a natural deduction rule is translated by the term $c^i$ of type $A^{Ku}$, where $c^i$ is an intuitionistic proof term of $A^{Ku}$.

\begin{definition}[Translation of terms]
Kuroda's translation is inductively defined on the terms of the \lpcm by:
\[
  \begin{array}{l}
        x^{Ku} \coloneqq x \\
        c^{Ku} \coloneqq \left\{
        \begin{array}{ll}
            \lambda p. ~\Prf ~(\neg\neg p) &\text{ if $c = \Prf$} \\
            \lambda x. ~\lambda p. ~\fa ~x ~(\lambda z. ~\neg\neg (p ~z)) &\text{ if $c = \fa$} \\
            c^i &\text{ if $c$ is a constant representing a natural deduction rule} \\
            c &\text{ otherwise} \\
        \end{array}
        \right.\\
        s^{Ku} \coloneqq s \\
        (\lambda x : A. ~t)^{Ku} \coloneqq \lambda x : A^{Ku}. ~t^{Ku} \\
        (\Pi x : A. ~B)^{Ku} \coloneqq \Pi x : A^{Ku}. ~B^{Ku} \\
        (t ~u)^{Ku} \coloneqq t^{Ku} ~u^{Ku} \\
  \end{array}
\]
\end{definition}

\begin{proposition}
\label{prop_const_Ku}
For every constant $c : A \in \Sigma_{HOL}$ representing a natural deduction rule, we have $\vdash c^i : A^{Ku}$ in the theory $(\Sigma_{HOL}^i, \R_{HOL})$.
\end{proposition}

\begin{proof}
We have formalized the proof terms $c^i$ in \dk\footnote{See \url{https://github.com/Deducteam/Construkti/blob/master/kuroda.dk}.}. For instance, $\mathsf{top_i}^i$ is given in \Cref{sec_implem}. 
\end{proof}

As we are not mixing sorts, propositions and proofs, we know that the symbol $\fa$, the symbol $\Prf$ and the constants representing the natural deduction rules only occur in the grammar $\kappa_3$. Therefore, any type $A \in \kappa_i$ is modified by Kuroda's translation for $i = 3$, whereas $A^{Ku} = A$ for $i \neq 3$.

We have defined the translation for terms, and we now want to define it for theories. Intuitively, we would like to translate a rewrite rule $\ell \lra r$ by $\ell^{Ku} \lra r^{Ku}$. However, if the head constant of $\ell$ is $\Prf$ or $\fa$, then the head symbol of $\ell^{Ku}$ is $\Prf^{Ku}$ or $\fa^{Ku}$, that is a $\lambda$-abstraction and not a constant. Hence $\ell^{Ku} \lra r^{Ku}$ may not be a valid rewrite rule in the \lpcm. We write $\lfloor \ell^{Ku} \rfloor$ for the term obtained by $\beta$-reducing the head symbol of $\ell^{Ku}$ if it is $\Prf^{Ku}$ or $\fa^{Ku}$.
\begin{definition}
The translation $t \mapsto t^{Ku}$ is extended to contexts, signatures and rewrite systems by:
\[
  \begin{array}{l}
  		\langle \rangle^{Ku} \Coloneqq \langle \rangle \\
        (\Gamma, x : A)^{Ku} \coloneqq \Gamma^{Ku}, x : A^{Ku} \\
        (\Sigma, c : A)^{Ku} \coloneqq \Sigma^{Ku}, c : A^{Ku} \\
        (\R, \ell \lra r)^{Ku} \coloneqq \R^{Ku}, \lfloor \ell^{Ku} \rfloor \lra r^{Ku} \\
  \end{array}
\]
\end{definition}
When translating a theory encoded in higher-order logic, we replace $\Sigma_{HOL}^c$ by $\Sigma_{HOL}^i$, and we translate the user-defined signature $\Sigma_{\T}$ and rewrite system $\R_{\T}$.

\begin{definition}[Translation of theories]
Let $\T = (\Sigma_{HOL}^c \cup \Sigma_{\T}, \R_{HOL} \cup \R_{\T})$ be a theory encoded in higher-order logic. The translation of $\T$ is $\T^{Ku} = (\Sigma_{HOL}^i \cup \Sigma_{\T}^{Ku}, \R_{HOL} \cup \R_{\T}^{Ku})$.
\end{definition}
Remark that $\T^{Ku}$ is a theory. Specifically, rewrite rules $\lfloor \ell^{Ku} \rfloor \lra r^{Ku} \in \R_{\T}^{Ku}$ are always well-defined, since $\ell$ is neither $\Prf$ nor $\fa$, and by definition of $\lfloor \ell^{Ku} \rfloor$.

\subsection{Embedding Classical Logic into Intuitionistic Logic}

We aim at proving that the extension of Kuroda's translation in the \lpcm indeed embeds classical logic into intuitionistic logic. In other words, we want to show that  $\Gamma \vdash t : A$ in $\T$ entails $\Gamma^{Ku} \vdash t^{Ku} : A^{Ku}$ in $\T^{Ku}$. To do so, we translate the derivations step by step. In particular, when the \textsc{Conv} rule is used with $A \equiv_{\beta\R} B$ in $\T$, we want to have $A^{Ku} \equiv_{\beta\R} B^{Ku}$ in $\T^{Ku}$.

\begin{lemma}[Translation of substitutions]
\label{subst_lp_Ku}
$(t[z \leftarrow w])^{Ku} = t^{Ku}[z \leftarrow w^{Ku}]$
\end{lemma}

\begin{proof}
By induction on the term $t$. We have $(c[z \leftarrow w])^{Ku} = c^{Ku} = c^{Ku}[z \leftarrow w^{Ku}]$ since $c^{Ku}$ is a closed term. Similarly, $(s[z \leftarrow w])^{Ku} = s^{Ku} = s^{Ku}[z \leftarrow w^{Ku}]$. If $x \neq z$, then $(x[z \leftarrow w])^{Ku} = x^{Ku} = x^{Ku}[z \leftarrow w^{Ku}]$. If $x = z$, then $(x[z \leftarrow w])^{Ku} = w^{Ku} = x[z \leftarrow w^{Ku}] = x^{Ku}[z \leftarrow w^{Ku}]$. The cases for $\lambda$-abstractions, dependent types, and applications follow from the induction hypotheses.
\end{proof}

\begin{lemma}[Translation of conversions]
\label{conv_lp_Ku}
If $A \equiv_{\beta\R} B$ in $\T$, then $A^{Ku} \equiv_{\beta\R} B^{Ku}$ in $\T^{Ku}$.
\end{lemma}

\begin{proof}
By induction on the construction of $A \equiv_{\beta\R} B$.
\begin{itemize}
    \item If $\ell \lra r$ in $\T$, then we show $(\ell\theta)^{Ku} \equiv_{\beta\R} (r\theta)^{Ku}$ in $\T^{Ku}$ for any substitution $\theta$. For $\ell \lra r \in \R_{HOL}$, we have $\ell^{Ku} = \ell$ and $r^{Ku} = r$, and we use \Cref{subst_lp_Ku}. For $\ell \lra r \in \R_{\T}$, we have $\lfloor \ell^{Ku} \rfloor \lra r^{Ku} \in \R_{\T}^{Ku}$, which entails that $(\ell\theta)^{Ku} = \ell^{Ku}\theta^{Ku} \equiv_{\beta\R} \lfloor \ell^{Ku} \rfloor \theta^{Ku} \equiv_{\beta\R} r^{Ku}\theta^{Ku} = (r\theta)^{Ku}$ by \Cref{subst_lp_Ku}.
    \item If $(\lambda x : A. ~t) ~u \lra t[x \leftarrow u]$ in $\T$, then we have $((\lambda x : A. ~t) ~u)^{Ku} = (\lambda x : A^{Ku}. ~t^{Ku}) ~u^{Ku}$ , which $\beta$-reduces to $t^{Ku}[x \leftarrow u^{Ku}]$, that is $(t[x \leftarrow u])^{Ku}$ using \Cref{subst_lp_Ku}.
    \item Closure by context, reflexivity, symmetry, and transitivity are immediate.
\end{itemize}
\end{proof}

\begin{theorem}[Translation of judgments]
Let $\T$ be a theory encoded in higher-order logic.
\begin{itemize}
    \item If $\vdash \Gamma$ in $\T$ then $\vdash \Gamma^{Ku}$ in $\T^{Ku}$.
    \item If $\Gamma \vdash t : A$ in $\T$ then $\Gamma^{Ku} \vdash t^{Ku} : A^{Ku}$ in $\T^{Ku}$.
\end{itemize}
\end{theorem}

\begin{proof}
We proceed by induction on the derivation. We present the most interesting cases, the others follow the definition and the induction hypotheses.
\begin{itemize}
\item \underline{\textsc{Const}}: By induction we have $\vdash \Gamma^{Ku}$ and $\Gamma^{Ku} \vdash A^{Ku} : s^{Ku}$ in $\T^{Ku}$. 

If $c : A \in \Sigma_{\T}$, then $c : A^{Ku} \in \Sigma_{\T}^{Ku}$ and we derive $\Gamma^{Ku} \vdash c : A^{Ku}$ using \textsc{Const}.

Suppose that $c = \Prf$. We simply derive $\Gamma^{Ku} \vdash \lambda p. ~\Prf ~(\neg\neg p) : \Prop \ra \Type$, that is $\Gamma^{Ku} \vdash \Prf^{Ku} : (\Prop \ra \Type)^{Ku}$, in $\T^{Ku}$.

Suppose that $c = \fa$. We simply derive $\Gamma^{Ku} \vdash \lambda x. ~\lambda p. ~\fa ~x ~(\lambda z. ~\neg\neg (p ~z)) : \Pi x : \Set. ~(\El ~x \ra \Prop) \ra \Prop$, that is $\Gamma^{Ku} \vdash \fa^{Ku} : (\Pi x : \Set. ~(\El ~x \ra \Prop) \ra \Prop)^{Ku}$, in $\T^{Ku}$.

Suppose that $c$ is a constant representing a natural deduction rule. Using \Cref{prop_const_Ku}, we have $\Gamma^{Ku} \vdash c^i : A^{Ku}$ in $\T^{Ku}$, that is $\Gamma^{Ku} \vdash c^{Ku} : A^{Ku}$. In particular, we replace the classical axiom $\mathsf{pem} : \Pi p : \Prop. ~\Prf ~(p \disj \neg p)$ by the intuitionistic term $\mathsf{pem}^i : \Pi p : \Prop. ~\Prf ~(\neg\neg (p \disj \neg p))$.

Otherwise, $c : A \in \Sigma_{HOL}$ but is not $\Prf$, not $\fa$, and not a constant representing a natural deduction rule. Then $A$ does not contain $\Prf$ and $\fa$, so $A^{Ku} = A$. We derive $\Gamma^{Ku} \vdash c : A^{Ku}$ using \textsc{Const}.

\item \underline{\textsc{Conv}}: By induction we have $\Gamma^{Ku} \vdash t^{Ku} : A^{Ku}$ in $\T^{Ku}$ and $\Gamma^{Ku} \vdash B^{Ku} : s^{Ku}$ in $\T^{Ku}$. From \Cref{conv_lp_Ku}, we know that $A^{Ku} \equiv_{\beta\R} B^{Ku}$, and we conclude that $\Gamma^{Ku} \vdash t^{Ku} : B^{Ku}$ in $\T^{Ku}$ using \textsc{Conv}.
\end{itemize}
\end{proof}

\begin{example}[Translated equational theory]
The translation of the theory $\T = (\Sigma_{HOL} \cup \Sigma_{eq}, \R_{HOL} \cup \R_{eq})$ of \Cref{ex_equality} is obtained by taking the equality symbol $= ~: \Pi a : \Set. ~\El ~a \ra \El ~a \ra \Prop$ (which remains unchanged), and by transforming the rewrite rule $\Prf ~(= a ~x ~y) \lra \Pi P : \El ~a \ra \Prop. ~\Prf ~(P ~x) \ra \Prf ~(P ~y)$ into $\Prf ~(\neg\neg (= a ~x ~y)) \lra \Pi P : \El ~a \ra \Prop. ~\Prf ~(\neg\neg (P ~x)) \ra \Prf ~(\neg\neg (P ~y))$. The proof of reflexivity is now given by $\lambda a : \Set. ~\mathsf{all_i}^i ~a ~(\lambda x : \El ~a. = a ~x ~x) ~(\lambda x : \El ~a. ~\lambda P : \El ~a \ra \Prop. ~\lambda P_x : \Prf ~(\neg\neg (P ~x)). ~P_x)$ which is of type $\Pi a : \Set. ~\Prf ~(\neg\neg (\fa ~a ~(\lambda x : \El ~a. ~\neg\neg (= a ~x ~x))))$. 
\end{example}

\subsection{Back to the Original Theory}

We have shown that, in the \lpcm, $\Gamma \vdash t : A$ in $\T$ implies $\Gamma^{Ku} \vdash t^{Ku} : A^{Ku}$ in $\T^{Ku}$. We now want to prove the reverse implication: if there exists an intuitionistic proof of $A^{Ku}$ in $\T^{Ku}$, then there exists a classical proof of $A$ in $\T$. To do so, we reason in two steps: first we show that it is possible to build a proof of $A$ from a proof of $A^{Ku}$ in classical logic, and then we show that any result in $\T^{Ku}$ can also be derived in $\T$.

The first step consists in proving that, for any $A \in \kappa_3$, it is possible to derive $A^{Ku}$ from $A$. For this, we show that any proposition $P$ and its translation $P^{Ku}$ are classically equivalent. Such a result is not necessarily true in higher-order logic. We assume some property, called the Kuroda equivalence.

\begin{definition}[Kuroda equivalence]
Let $\Gamma$ be a context, $t$ be a constant or a variable such that $\Gamma \vdash t : T_1 \ra \ldots \ra T_n \ra \Prop$, and $u_1, \ldots, u_n$ be terms such that $\Gamma \vdash u_i : T_i$. There exists some $p$ such that $\Gamma \vdash p : \Prf ~((t ~u_1 ~\ldots ~u_n)^{Ku} \eqi t ~u_1 ~\ldots ~u_n)$.
\end{definition}
The Kuroda equivalence property is derivable from functional extensionality and propositional extensionality in classical logic~\cite{kuroda_hol}. Remark that it is satisfied for the usual logical connectives and quantifiers. For instance, we have $A_{Ku} \conj B_{Ku} \eqi A \conj B$ and $\fa x ~\neg\neg A_{Ku} \eqi \fa x ~A$ in classical logic. In the rest of this paper, we work assuming the Kuroda equivalence.

\begin{lemma}
\label{lemma_prop_cases}
Any proposition $P$ is $\beta$-convertible to a variable $x$, a constant $c$, or an application $t ~u_1 \ldots ~u_n$ where $t$ is a constant or a variable of type $T_1 \ra \ldots \ra T_n \ra \Prop$ and $u_1, \ldots, u_n$ are terms of type $T_1, \ldots, T_n$.
\end{lemma}
The constant $c$ may be $\top$ or $\bot$, and the head symbol of the application may be any connective, quantifier or predicate.

\begin{proposition}
\label{prop_translate}
Let $\Gamma \vdash P : \Prop$. In the theory $(\Sigma_{HOL}^c \cup \Sigma, \R_{HOL} \cup \R)$, there exists some proof term $m_P$ such that $\Gamma \vdash m_P : \Prf ~ (P^{Ku} \eqi P)$.
\end{proposition}

\begin{proof}
We distinguish cases thanks to \Cref{lemma_prop_cases}. 
\begin{itemize}
\item Suppose that $P$ is $\beta$-convertible to a variable $x$. We have $x^{Ku} = x$ so we build some $m_x$ such that $\Gamma \vdash m_x : \Prf ~(x^{Ku} \eqi x)$. Since $P$ is $\beta$-convertible to $x$, $P^{Ku}$ is $\beta$-convertible to $x^{Ku}$ (see \Cref{conv_lp_Ku}) and we conclude that $\Gamma \vdash m_x : \Prf ~(P^{Ku} \eqi P)$.
\item If $P$ is $\beta$-convertible to a constant $c$, then we are in the case where $c^{Ku} = c$ and we proceed similarly.
\item Suppose that $P$ is $\beta$-convertible to an application $t ~u_1 \ldots ~u_n$ where $t$ is a constant or a variable. $P^{Ku}$ is $\beta$-convertible to $(t ~u_1 \ldots ~u_n)^{Ku}$ and we conclude using the Kuroda equivalence.
\end{itemize}
\end{proof}

\begin{lemma}
\label{subformula_translate}
Let $A \in \kappa_3$ and $\ell$ be a strict subterm of $A$. In the theory $(\Sigma_{HOL}^c \cup \Sigma, \R_{HOL} \cup \R)$, for any context $\Gamma$, there exists some $t$ such that $\Gamma \vdash t : A[\ell]$ if and only if there exists some $t'$ such that $\Gamma \vdash t' : A[\ell^{Ku}]$.
\end{lemma}

\begin{proof}
We proceed by induction on the term $A$ using the fact that $A$ is generated by $\kappa_3$.
\begin{itemize}
\item Suppose that $A = \Prf ~P$. If $\fa$ does not occur in $\ell$, then $\ell^{Ku} = \ell$ and $P[\ell^{Ku}] = P[\ell]$, so we directly conclude. Otherwise, we use \Cref{prop_translate} on the right proposition.

\item Suppose that $A = \Pi x : B. ~C$ with $B \in \kappa_1$ or $B \in \kappa_2$. If $\ell$ occurs in $B$, then by definition $B[\ell^{Ku}] = B[\ell]$, so $\ell^{Ku} = \ell$ and we directly conclude. Suppose that $\ell$ only occurs in $C$ and that there exists some $t$ such that $\Gamma \vdash t : \Pi x : B. ~C[\ell]$. By induction on $C$ with $\Gamma, x : B \vdash t ~x : C[\ell]$ (obtained by weakening), we get some $t_C'$ such that $\Gamma, x : B \vdash t_C' : C[\ell^{Ku}]$. Therefore, we have $\Gamma \vdash \lambda x : B. ~t_C' : \Pi x : B. ~C[\ell^{Ku}]$. We proceed similarly for the reverse implication.

\item Suppose that $A = B \ra C$ with $B,C \in \kappa_3$. Suppose that we have $\Gamma \vdash t : B[\ell] \ra C[\ell]$. By induction on $B$ with $\Gamma, x : B[\ell^{Ku}] \vdash x : B[\ell^{Ku}]$, we get some $t_B$ such that $\Gamma, x : B[\ell^{Ku}] \vdash t_B : B[\ell]$. By induction on $C$ with $\Gamma, x : B[\ell^{Ku}] \vdash t ~t_B : C[\ell]$, we get some $t_C'$ such that $\Gamma, x : B[\ell^{Ku}] \vdash t_C' : C[\ell^{Ku}]$. We conclude that $\Gamma \vdash \lambda x : B[\ell^{Ku}]. ~t_C' : B[\ell^{Ku}] \ra C[\ell^{Ku}]$. We proceed similarly for the reverse implication.
\end{itemize}
\end{proof}

\begin{lemma}
\label{formula_translate}
Let $A \in \kappa_3$. In the theory $(\Sigma_{HOL}^c \cup \Sigma, \R_{HOL} \cup \R)$, for any context $\Gamma$, there exists some $t$ such that $\Gamma \vdash t : A$ if and only if there exists some $t'$ such that $\Gamma \vdash t' : A^{Ku}$.
\end{lemma}

\begin{proof}
We proceed by induction on the term $A$ using the fact that $A$ is generated by $\kappa_3$. We use \Cref{subformula_translate} and the double-negation elimination. 
\end{proof}

We have shown that it is possible to build a proof of $A$ in $\T^{Ku}$ using a proof of $A^{Ku}$ and the principle of excluded middle. The next step is to derive a proof of $A$ in the original theory $\T$. In particular, it requires to replace each use of $\lfloor \ell^{Ku} \rfloor \lra r^{Ku} \in \R^{Ku}_{\T}$ by a use of $\ell \lra r \in \R_{\T}$.

\begin{lemma}
\label{prop_rule_Ku}
Let $A \in \kappa_3$ such that $\Gamma \vdash t : A[\ell^{Ku}]$. Using $\ell \lra r$, there exists some $t'$ such that $\Gamma \vdash t' : A[r^{Ku}]$.
\end{lemma}

\begin{proof}
Using \Cref{subformula_translate}, there exists some $t'$ such that $\Gamma \vdash t' : A[\ell]$. Using $\ell \lra r$, we have $\Gamma \vdash t' : A[r]$. We use \Cref{subformula_translate} to obtain some $t''$ such that $\Gamma \vdash t'' : A[r^{Ku}]$.
\end{proof}

\begin{lemma}
\label{prop_rules_Ku}
Let $(\Sigma_{HOL}^c \cup \Sigma, \R_{HOL} \cup \R^{Ku})$ and $(\Sigma_{HOL}^c \cup \Sigma, \R_{HOL} \cup \R)$ be two theories, abbreviated $\R^{Ku}$ and $\R$.
\begin{itemize}
\item If $\vdash \Gamma$ in $\R^{Ku}$ then $\vdash \Gamma$ in $\R$. 
\item If $\Gamma \vdash t : A$ in $\R^{Ku}$ and $A \in \kappa_i$ with $i \in \{ 1, 2, 4, 5 \}$, then $\Gamma \vdash t : A$ in $\R$.
\item If $\Gamma \vdash t : A$ in $\R^{Ku}$ and $A \in \kappa_3$, then there exists some $t'$ such that $\Gamma \vdash t' : A$ in $\R$.
\end{itemize}
\end{lemma}

\begin{proof}
We proceed by induction on the typing derivation. We only present the relevant cases.
\begin{itemize}
\item \underline{\textsc{Abs}}: Suppose that $\Gamma \vdash A : \Type$ and $\Gamma, x : A \vdash B : s$ and $\Gamma, x : A \vdash t : B$ in $\R^{Ku}$. By induction we have $\Gamma \vdash A : \Type$ and $\Gamma, x : A \vdash B : s$ in $\R$. 

If $B \in \kappa_i$ with $i \in \{ 1, 2 \}$, then by induction we have $\Gamma, x : A \vdash t : B$ in $\R$, and we derive $\Gamma \vdash \lambda x : A. ~t : \Pi x : A. ~B$ in $\R$.

If $B \in \kappa_3$, then by induction we have $\Gamma, x : A \vdash t' : B$ in $\R$. We derive $\Gamma \vdash \lambda x : A. ~t' : \Pi x : A. ~B$ in $\R$.

\item \underline{\textsc{App}}: Suppose that $\Gamma \vdash t : \Pi x : A. ~B$ and $\Gamma \vdash u : A$ in $\R^{Ku}$. 

If $\Pi x : A. ~B \in \kappa_i$ with $i \in \{ 1, 2, 4 \}$, then by induction we have $\Gamma \vdash t : \Pi x : A. ~B$ and $\Gamma \vdash u : A$ in $\R$. We derive $\Gamma \vdash t ~u : B[x \leftarrow u]$ in $\R$.

If $\Pi x : A. ~B \in \kappa_3$, then by induction we have $\Gamma \vdash t' : \Pi x : A. ~B$ in $\R$. If $A \in \kappa_i$ with $i \in \{ 1, 2 \}$, then by induction we have $\Gamma \vdash u : A$ in $\R$, and we derive $\Gamma \vdash t' ~u : B[x \leftarrow u]$ in $\R$. If $A \in \kappa_3$ ($x$ does not occur in $B$), then by induction we have $\Gamma \vdash u' : A$ in $\R$, and we conclude that $\Gamma \vdash_c t' ~u' : B$.

\item \underline{\textsc{Conv}}: If $A \equiv_{\beta\R} B$ is obtained using $\beta$-conversion or the rewrite rules of $\R_{HOL}$, then we conclude using the induction hypothesis and the \textsc{Conv} rule. Otherwise, and without loss of generality, we consider that we only use one rewrite rule of $\R^{Ku}$ per \textsc{Conv} rule.

Suppose that $A \equiv_{\beta\R} B$ is obtained using the rewrite rule $\ell^{Ku} \lra r^{Ku} \in \R^{Ku}$. In that case, we have $A = C[\ell^{Ku}]$ and $B = C[r^{Ku}]$ (the case $A = C[r^{Ku}]$ and $B = C[\ell^{Ku}]$ is treated similarly). By assumption, we have $\Gamma \vdash t : C[\ell^{Ku}]$ and $\Gamma \vdash C[r^{Ku}] : s$ in $R^{Ku}$. 

If $A,B \in \kappa_i$ with $i \in \{ 1, 2, 4, 5 \}$, then $\ell^{Ku} = \ell$ and $r^{Ku} = r$. By induction we have $\Gamma \vdash t : C[\ell^{Ku}]$ and $\Gamma \vdash C[r^{Ku}] : s$ in $\R$. We apply \textsc{Conv} with $C[\ell] \equiv_{\beta\R} C[r]$.

If $A,B \in \kappa_3$, then by induction we have $\Gamma \vdash t' : C[\ell^{Ku}]$ and $\Gamma \vdash C[r^{Ku}] : s$ in $\R$. We conclude using \Cref{prop_rule_Ku}.
\end{itemize}
\end{proof}

We now have all the tools to show that, for any intuitionistic proof of $A^{Ku}$ in the translated theory $\T^{Ku}$, there exists a classical proof of $A$ in the original theory $\T$.

\begin{theorem}
Let $\T$ be a theory encoded in higher-order logic and $A \in \kappa_3$. If $\Gamma^{Ku} \vdash t : A^{Ku}$ in $\T^{Ku}$, then under the Kuroda equivalence there exists some term $t'$ such that $\Gamma \vdash t' : A$ in $\T$.
\end{theorem}

\begin{proof}
We directly have $\Gamma^{Ku} \vdash t : A^{Ku}$ in $(\Sigma_{HOL}^c \cup \Sigma_{\T}^{Ku}, \R_{HOL} \cup \R_{\T}^{Ku})$. 
\begin{itemize}
\item By \Cref{formula_translate}, there exists some $t'$ such that $\Gamma^{Ku} \vdash t' : A$ in $(\Sigma_{HOL}^c \cup \Sigma_{\T}^{Ku}, \R_{HOL} \cup \R_{\T}^{Ku})$ and under the Kuroda equivalence. 
\item Using \Cref{prop_rules_Ku}, there exists some $t''$ such that $\Gamma^{Ku} \vdash t'' : A$ in $(\Sigma_{HOL}^c \cup \Sigma_{\T}^{Ku}, \R_{HOL} \cup \R_{\T})$.
\item We replace the signature $\Sigma_{\T}^{Ku}$ by $\Sigma_{\T}$. For each constant $c : C \in \Sigma_{\T}$ with $C \in \kappa_3$, we replace $c$ by $t_c$ (provided by \Cref{formula_translate}) in $t''$. We obtain $\Gamma^{Ku} \vdash t''[c \leftarrow t_c] : A$ in $(\Sigma_{HOL}^c \cup \Sigma_{\T}, \R_{HOL} \cup \R_{\T})$, that is in $\T$. These substitutions work since $c$ cannot occur in a dependent type.  
\item We replace the context $\Gamma^{Ku}$ by $\Gamma$. For each variable $x : B \in \Gamma$ with $B \in \kappa_3$, we replace $x$ by $t_x$ (provided by \Cref{formula_translate}) in $t''[c \leftarrow t_c]$. We obtain $\Gamma \vdash t''[c \leftarrow t_c][x \leftarrow t_x] : A$ in $\T$, which achieves the proof.
\end{itemize}
\end{proof}
The extension of Kuroda's translation to the \lpcm is a generalization of Brown and Rizkallah's translation for simple type theory~\cite{brown_rizkallah}. Indeed, if $\R_{\T} = \langle \rangle$, then we obtain the result in higher-order logic, at the only difference that proofs are represented by terms.

\section{Construkti, an Implementation for Dedukti Proofs}
\label{sec_implem}

\paragraph{Dedukti.} The \lpcm has been implemented in the \dk proof language. Abstractions $\lambda x : A. ~t$ are represented by \texttt{x : A => t}, and dependent types $\Pi x : A. ~B$ are represented by \texttt{x : A -> B}. Constants $c : A$ are specified by \texttt{c : A}, prefixed with the keyword \texttt{def} if the constant can be defined using rewrite rules. Rewrite rules $\ell \lra r$, where $x$ and $y$ are the free variables of $\ell$ and $r$, are represented by \texttt{[x,y] l --> r}. For instance, using the encoding  of the notions of proposition and proof, we can encode the addition on natural numbers via rewrite rules.

\begin{lstlisting}
nat : Set.
0 : El nat.
S : El nat -> El nat.
def add : El nat -> El nat -> El nat.
[x] add x 0 --> x.
[x, y] add x (S y) --> S (add x y).
\end{lstlisting}
Theorems are represented by \texttt{thm n : T := p}, where \texttt{n} is its name, \texttt{T} its statement and \texttt{p} its proof term. For checking that \texttt{p} is indeed a proof of \texttt{T}, we can use one of the type checkers of \dk, for instance \textsc{DKCheck}~\cite{saillard} or \textsc{Lambdapi}~\cite{deduktiengine}.

\paragraph{Construkti.} We have implemented \ck\footnote{Available at \url{https://github.com/Deducteam/Construkti}.}, a tool that performs Kuroda's translation on \dk proofs. \ck takes as input a \dk file containing the specification of a user-defined theory encoded in higher-order logic, as well as proofs in this theory. It returns a \dk file containing the specification of the translated theory, as well as the translated proofs. 

In this implementation, we insert one double negation after every $\Prf$ and $\fa$ symbols, and we replace the constants $c$ representing natural deduction rules by the terms $c^i$. For instance, the constant $\mathsf{top_i}$ of type $\Prf ~\top$, representing the introduction of tautology, is replaced in the formal proofs by the term $\mathsf{top_i}^i$ of type $\Prf ~(\neg\neg \top)$. The proof term $\mathsf{top_i}^i$ relies on the proof of $\Pi p : \Prop. ~\Prf ~(p \imp \neg\neg p)$.

\begin{lstlisting}
top_i : Prf top.

thm prop_double_neg : p : Prop -> Prf (imp p (not (not p))) 
:= p => imp_i p (not (not p)) 
	(pP => neg_i (not p) (pNP => neg_e p pNP pP)).

thm top_i_i : Prf (not (not top))
:= imp_e top (not (not top)) (prop_double_neg top) top_i.
\end{lstlisting}
So as to obtain readable theorems, we directly $\beta$-reduce every application of $\Prf^{Ku}$ and $\fa^{Ku}$.

\paragraph{Benchmark.} We have tested \ck on a benchmark of 101 \dk proofs, available in the file \texttt{hol-lib.dk}. These proofs encompass results related to connectives and quantifiers, classical formulas, De Morgan's laws, polymorphic equality, and basic arithmetic. The proofs are expressed in propositional, first-order and higher-order logics. This library of proofs includes user-defined rewrite rules---a feature of the \lpcm---and inference rules---thanks to the encoding of the notions of proposition and proof. We compare in \Cref{fig_table} the different characteristics of the library: the number of proofs, the number of classical proofs, the number of results expressed in higher-order logic, and the number of results that are expressed via admissible inference rules. 

\begin{table}[ht]
\centering
\begin{tabular}{|l|c|c|c|c|}
\hline
Content of & \multicolumn{4}{|c|}{Number of ...} \\\cline{2-5}
the library & proofs & classical proofs & higher-order results & admissible inference rules \\
\hline
Basic logic & 38 & 0 & 15 & 26 \\
Classical results & 12 & 12 & 9 & 3\\
De Morgan & 8 & 6 & 4 & 8 \\
Equality & 10 & 0 & 6 & 4 \\
Arithmetic & 33 & 0 & 0 & 16 \\
\hline
All & 101 & 18 & 34 & 57 \\
\hline
\end{tabular}
\caption{Comparison of the different libraries.}
\label{fig_table}
\end{table}
After running \ck, all the translated proofs of the translated theorems typecheck, and are expressed in intuitionistic logic.

\section{Conclusion}

In this paper, we have extended Kuroda's translation to the theories encoded in higher-logic in the \lpcm, that is $\lambda$-calculus extended with dependent types and user-defined rewrite rules. In this logical framework, proofs are terms following the Curry-Howard correspondence, and have to be effectively translated. Due to the encoding of the notions of proposition and proof in the \lpcm, we can assume, prove, and translate inference rules. We have implemented \ck, a tool that transforms \dk proofs following Kuroda's translation. Both \dk and \ck pave the way for interoperability between classical proof systems---such as \textsc{HOL Light} or \textsc{Mizar}---and intuitionistic proof systems---such as \textsc{Coq}, \textsc{Lean} or \textsc{Agda}.

\paragraph{Future work.} There exist large libraries of proofs in higher-order logic, for instance the \textsc{HOL Light} standard library. Blanqui~\cite{hol2dk} recently translated it to \textsc{Coq} via \dk, taking the excluded middle as an axiom. Future work would be to obtain an intuitionistic version of the \textsc{HOL Light} standard library, by applying Kuroda's translation and \ck.

\paragraph{Related work.} Double-negation translations aim at embedding classical logic into intuitionistic logic. As such, double-negation translations \textit{always} transform classical proofs into intuitionistic ones, but they modify the formulas during the process. Proof constructivization aims at transforming classical proofs into intuitionistic ones \textit{without} translating the formulas, but such a process does not necessarily succeed. Cauderlier~\cite{cauderlier} developed heuristics to constructivize proofs in \dk, via rewrite systems that try to remove instances of the principle of excluded middle or of the double-negation elimination. Gilbert~\cite{gilbert} designed a constructivization algorithm for first-order logic, that was tested in \dk and works in practice for large fragments of first-order logic.

\section*{Acknowledgments}

The author would like to thank Marc Aiguier, Gilles Dowek and Olivier Hermant for helpful discussions and valuable remarks about this work.

\nocite{*}
\bibliographystyle{eptcs}
\bibliography{biblio}
\end{document}